\documentclass{article}

\usepackage{amsmath,amssymb,amsthm}
\usepackage{caption}
\usepackage{subcaption}
\usepackage{graphicx}
\usepackage{float}
\usepackage{a4wide}
\usepackage{epstopdf}
\usepackage{url}
\usepackage{cite}

\newtheorem{theorem}{Theorem}

\newtheorem{lemma}[theorem]{Lemma}

\newtheorem{remark}[theorem]{Remark}

\usepackage{float}

\newcommand\blfootnote[1]{%
	\begingroup
	\renewcommand\thefootnote{}\footnote{#1}%
	\addtocounter{footnote}{-1}%
	\endgroup
}

\begin{document}
	
\title{A model for the dynamics of COVID-19 infection transmission 
in human with latent delay\blfootnote{This is a preprint 
of a paper whose final and definite form is published in 
'Afrika Matematika' at {\tt https://doi.org/10.1007/s13370-024-01226-0}.
Submitted 02-07-2022; Revised 12-09-2024; Accepted 14-12-2024.}}

\author{Amar N. Chatterjee$^{1}$\\ {\tt anchaterji@gmail.com}
\and Teklebirhan Abraha$^{2}$\\ {\tt tekbir98@yahoo.com}
\and Fahad Al Basir$^{3}$\\ {\tt fahadbasir@gmail.com}
\and Delfim F. M. Torres$^{4}$\thanks{Corresponding author. 
Email: delfim@ua.pt}\\ {\tt delfim@ua.pt}}
	
\date{$^{1}$\small{\mbox{Department of Mathematics, K.L.S. College, 
Nawada, Magadh University, Bodhgaya, Bihar-805110, India}}\\[0.2cm]
$^{2}$\small{Department of Mathematics, 
Aksum University, Aksum, Ethiopia}\\[0.2cm]
$^{3}$\small{Department of Mathematics, Asansol Girls' College, 
West Bengal 713304, India}\\[0.2cm]
$^{4}$\small{R\&D Unit CIDMA, Department of Mathematics, 
University of Aveiro, 3810-193 Aveiro, Portugal}}
	
\maketitle
	
% --------------------------------------------
	
\begin{abstract}
In this research, we have derived a mathematical model for within human dynamics 
of COVID-19 infection using delay differential equations. The new model considers a 'latent period' 
and 'the time for immune response' as delay parameters, allowing us to study the effects 
of time delays in human COVID-19 infection. We have determined the equilibrium points and analyzed 
their stability. The disease-free equilibrium is stable when the basic reproduction number, $R_0$, 
is below unity. Stability switch of the endemic equilibrium occurs through Hopf-bifurcation. 
This study shows that the effect of latent delay is stabilizing whereas 
immune response delay has a destabilizing nature.

\bigskip
		

\medskip

\end{abstract}
	
% -------------------------------------------
	
\section{Introduction}

The world is still suffering from the ongoing pandemic coronavirus disease (COVID-19)
caused by the SARS-CoV-2 virus. The outbreak of COVID-19 started in Wuhan China,
in December 2019, and now it has been spread over more than 226 countries and territories.
COVID-19 is a rapidly spreader infectious disease that threatens the health system of mankind. 
Infected people experience mild and moderate symptoms and recover without special treatment. 
However, the worst condition appears for the patient who is suffering from chronic diseases 
like heart, kidney, and lungs disease. These people become seriously ill 
and they need medical attention like oxygen support.

The virus is spread from infected person's mouth when they cough, sneeze and speak through droplets. 
By touching a contaminated surface or by breathing near some COVID-19 infected person, one may be 
infected. The crowded and indoor environment helps the virus for community's spreading. 
The World Health Organization (WHO) has approved several COVID-19 vaccines and the first 
mass vaccination program started in early December 2020. The listed vaccines approved by WHO 
are Pfizer/BioNTech Comirnaty, SII/COVISHIELD, AstraZeneca/AZD1222 vaccines, Janssen/Ad26.COV 2.S, 
Moderna COVID-19 vaccine (mRNA 1273), Sinopharm COVID-19 vaccine,
Sinovac-CoronaVac, and COVAXIN vaccine \cite{CDC:vaccines}.

In spite of vaccination, some people are being infected by different variants of the SARS-CoV-2 virus. 
Thus, proper non-pharmaceutical intervention and maintaining Standard Operating Procedures (SOP)
can only reduce the chances of reinfection. Recently, The Defence Research Development Organisation 
(DRDO, India) gave the license to Granules India to manufacture COVID-19 treatment drug, 
2-Deoxy-D-Glucose (2-DG). It is an antiviral and anti-inflammatory drug (see Figure~\ref{fig1}).
% --------------------------
\begin{figure}
\centering
\includegraphics[scale=0.4]{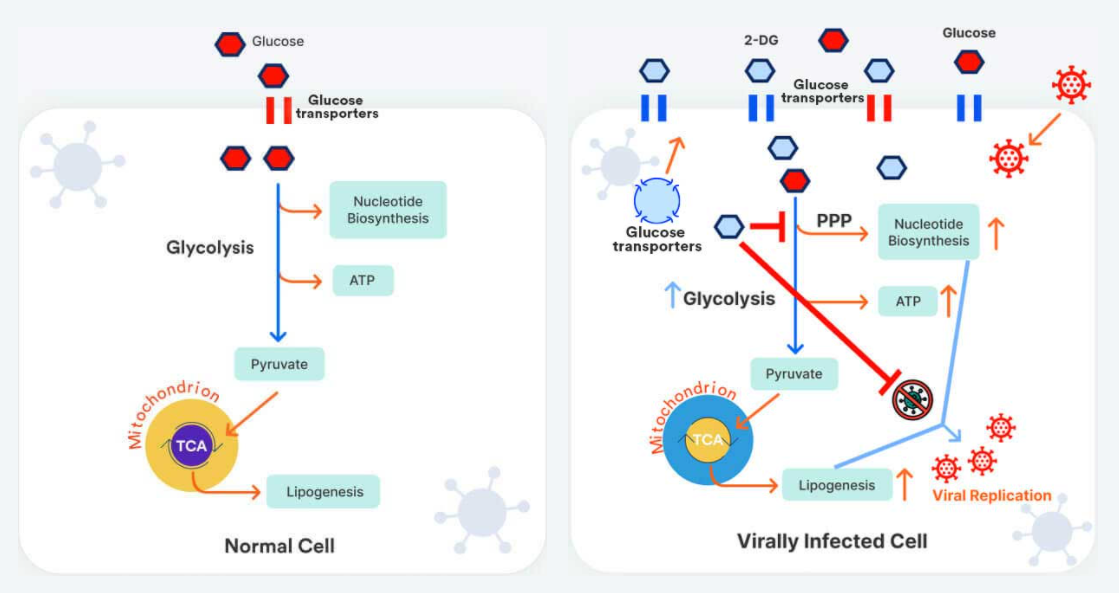}
\caption{Role of 2-DG -- an antiviral effect as it acts on host cells. Source: \cite{2dg}.}
\label{fig1}
\end{figure}
% --------------------------
Based on Phase II and Phase III trials, 2-DG has received emergency use authorization
to be a useful adjunct therapy in moderate COVID-19 patients in the hospital settings.

In the health sector, mathematical models are used to understand the disease dynamics
and help for decision-making for controlling the infection \cite{book:Covid-19,PKT1,PKT2}.
Currently, there are many mathematical models used to explain the disease progression
of the COVID-19 pandemic, for more details we refer the reader, e.g., to
\cite{PKT3,PKT4,Samui,Mondal1,Rahman,tang2017novel}. When the COVID-19 has arrived, 
researchers mainly working on mathematical biology have formulated mathematical models 
to study the spread of the pandemic, prevention and control of the pandemic, 
and the influence of prevention of measures 
\cite{PKT5,Zine,Mahrouf,Ndairou01,Silvaetal,Mondal3,Cantin}.
The mathematical study of intrahost viral dynamics of COVID-19 infection dynamics
plays a pivotal role to develop drugs or vaccines. Also, these studies help us to find the role 
of the existing drugs against COVID-19 infection. Mathematical modeling with real data contributes 
to widely discover the dynamical characteristics of an infection at cellular level 
\cite{Mondal1,qureshi2019fractional,qureshi2019transmission,Roy2,Chatterjee03,Roy1,Ghosh,MyID:480,Wang}.

Tang, Ma and Bai developed a model of virus and host immune response
dynamics, incorporating the concentration of DPP4 receptors for MARS-CoV infection \cite{tang2017novel}. 
Chatterjee and Basir \cite{Amar}, and Mondal, Samui and Chatterjee \cite{Mondal3} studied a mathematical 
model indicating the dynamical activities of epithelial cells during SARS-CoV-2 infection in the presence 
of cytotoxic T lymphocytes (CTL) response. Hernandez-Vargas and Velasco-Hernandeza described a model 
in cellular level dynamics of SARS-CoV-2 and T cell responses against the viral replication during 
the COVID-19 infection \cite{hernandez2020host}. The effect of the several pathogenic characteristics 
of SARS-CoV-2 and host immune response is studied by Wang et al. \cite{wang2020modeling}. They also 
evaluated and observed that anti-inflammatory treatment strategies, or a combination of antiviral 
drugs with interferon, are effective in reducing the plateau phase in viral load \cite{wang2020modeling}. 
Chatterjee and co-authors proposed a fractional differential equations model in cellular level,
accounting the lytic and non-lytic effects of immune response in the kinetics of the model, 
exploring the effect of a commonly used antiviral drug in COVID-19 treatment, while applying 
an optimal control-theoretic approach \cite{Chatterjee1,Chatterjee2}. 
In the previous models of COVID-19, time lag was not considered by the researchers. 
Here, we study the effects of time delays in human COVID-19 infection.

In SARS-CoV-2 infection, the virus life cycle plays a crucial role in disease progression.
The binding of a viral particle to a receptor on a target epithelial cell initiates
a series of events that can ultimately lead to the target cell becoming productively infected, i.e., 
producing new virus. In the previous models, it is assumed that the process
of infection is instantaneous. In other words, it is assumed that as soon as the virus
contacts a target cell, then the cell begins producing virus. However, in reality,
there is a time lag between initial viral entry into a cell and the time the cell
become virus producing. For this reason, in this work we consider this delay
to determine accurately the half life of free virus from drug perturbation experiments.
We will show that the delay affects the estimated value for the infected T-cell loss rate
when one assumes the drug is not completely effective.

We incorporate the intracellular delay in the model proposed in \cite{Amar}
by assuming a latent delay $\tau_1>0$ that means the generation of virus producing cells 
at time $t$ is due to the infection of target cells at time $(t-\tau_1)$. We also consider 
an immune response delay $\tau_2>0$ to account for the time needed to activate CTL responses, 
i.e., the CTL response against pathogens at time $t$ were activated by the infected 
epithelial cells at time $(t-\tau_2)$.

Our study is organized in the following manner. In Section~\ref{sec2}, a compartmental model 
of intrahost immune responses on SARS-CoV-2 viral dynamics is proposed. Then, we modify the 
proposed model taking into account the presence of delays. Section~\ref{sec3} deals with 
the analysis of the model, such as non-negativity and boundedness of all the solutions 
of the system. In Section~\ref{sec4}, the delayed induced mathematical model is analysed.
Numerical simulations are performed in Section~\ref{sec5} to examine whether the analytical 
results are associated with the numerical findings or not. Finally, we discuss about 
our inclusive analytical and numerical results in Section~\ref{sec6}.

% -------------------------------------------

\section{Model derivation}
\label{sec2}

The base model was proposed by Chatterjee and Al Basir \cite{Amar}.
Here, we extend the model in \cite{Amar} using a saturated infection
rate and two time delay factors.

Chatterjee and Al Basir \cite{Amar} proposed the model considering
the interaction between epithelial cells and SARS-CoV-2 virus along
with CTL responses over the infected cells. Five populations, namely,
the uninfected epithelial cells $T(t)$, infected cells $I(t)$,
angiotensin-converting enzyme 2 (ACE2) receptor of the epithelial
cells $E(t)$, SARS-CoV-2 virus $V(t)$
and CTLs against the pathogen, denoted as $C(t)$, were considered.
Mathematically, the model by Chatterjee and Al Basir is as follows:	
\begin{equation}
\label{mod1}
\begin{cases}
\displaystyle \frac{dT}{dt}=\lambda_1 - \beta E V T - d_T T,\\[0.3cm]
\displaystyle \frac{dI}{dt}=\beta E V T - d_I I - p I C,\\[0.3cm]
\displaystyle \frac{dV}{dt}= m d_I I - d_V V,\\[0.3cm]
\displaystyle \frac{dE}{dt}= \lambda_2 - \theta \beta E V T - d_E E,\\[0.3cm]
\displaystyle \frac{dC}{dt}=\alpha I C \left(1 - \frac{C}{C_{\max}}\right) - d_c C.
\end{cases}
\end{equation}
It is assumed that the susceptible epithelial cells are produced at a rate
$\lambda_1$ from the precursor cells and die at a rate $d_T$. The susceptible
cells become infected at a rate $\beta E(t) V(t) T(t)$. The infection rate
is inhibited by the nonlytic effect of CTLs at a rate $1+qC$, where $q$ is
the efficacy of the nonlytic effect. The constant $d_I$ is the death rate
of the infected epithelial cells. Infected cells are also cleared by the lytic
effect of the body's defensive CTLs at a rate $p I C$, where $p$ is the efficacy
of lytic effect. The infected cells produce new viruses at the rate $m d_I$ during their life;
and $d_V$ is the death rate of new virions, where $m$ is any positive integer. It is also assumed
that ACE2 is produced from the surface of uninfected epithelial cells at the constant rate $\lambda_2$
and the ACE2 is destroyed, when free viruses try to infect uninfected cells, at the rate
$\theta \beta E(t) V(t) T(t)$ and is hydrolyzed at the rate $d_E E$.
The CTL proliferation, in the presence of infected cells, is described by the term
\[
\alpha I C\left(1-\frac{C}{C_{\max}}\right),
\]
which shows the antigen-dependent proliferation. Here, we consider the logistic
growth of CTL with $C_{\max}$ as the maximum concentration of CTL and $d_C$ 
its rate of decay. We now make the following additional considerations to derive the desired model.

As discussed in the Introduction, we incorporate the intracellular delay in the model 
given by \eqref{mod1} by assuming that the generation of virus producing cells at time $t$ is due to
the infection of target cells at time $t-\tau_1$, where $\tau_1$ is a constant.
Here $e^{-d_I\tau_1}$ represents the survival probability during the latent time $\tau_1$.
	
We also consider an immune delay $\tau_2$ to account for the time needed to activate
CTL responses, i.e., the CTL response against pathogens at time $t$ were activated
by the infected epithelial cells at time $t-\tau_2$, where $\tau_2$ is constant.
	
The equations describing this new model are then given by
\begin{equation}
\label{mod2}
\begin{cases}
\displaystyle \frac{dT}{dt}=\lambda_1 - (1-\epsilon_1)\frac{\beta E V T}{1+q C}  - d_T T,\\[0.3cm]
\frac{dI}{dt}=(1-\epsilon_1)\frac{e^{-d_I\tau_1}\beta E(t-\tau_1)
\displaystyle V(t-\tau_1) T(t-\tau_1)}{1+q C(t-\tau_1)} - d_I I - p I C,\\[0.3cm]
\displaystyle \frac{dV}{dt}= (1-\epsilon_2)m d_I I - d_V V,\\[0.3cm]
\displaystyle \frac{dE}{dt}= \lambda_2 - \theta \beta E V T - d_E E,\\[0.3cm]
\displaystyle \frac{dC}{dt}=\alpha I(t-\tau_2)
C(t-\tau_2) \left(1 - \frac{C(t-\tau_2)}{C_{\max}}\right) - d_c C,
\end{cases}
\end{equation}
where $\tau_1$ and $\tau_2$ are the latent and the immune delays, respectively.
A short description of the model parameters and their values
is shown in Table~\ref{table1}.
% -----------------------------------------
\begin{table}[H]
\begin{center}
\caption{Description of the parameters of the proposed model \eqref{mod2}
with the values used for numerical simulations.}
\begin{tabular}{|c|c|c|} \hline
Parameters  & Description  & Value \\ \hline
$\lambda_1$ & Production rate of uninfected cell & 5   \\
$\lambda_2$ & Production rate of ACE2 & 1  \\
$\beta$ & Disease transmission rate& 0.0001 \\
$\theta$ & Bonding rate of ACE2& 0.2 \\
$d_T$ & Death rate of uninfected cells & 0.1 \\
$d_I$ & Death rate of infected cells & 0.1 \\
$d_V$ & Removal rate of virus &  0.1   \\
$d_E$ & Hydrolyzing rate of epithelial cells & 0.1 \\
$d_c$ & Decay rate of CTL & 0.1 \\
$p$&Killing rate of infected cells by CTL&0.01\\
$m$&Number of new virions produced&20\\
$\alpha$&Proliferation rate of CTL&0.22\\
$C_{\max}$ & Maximum proliferation of CTL & 100 \\ \hline
\end{tabular}
\label{table1}
\end{center}
\end{table}

% -------------------------------------------

\section{Analysis of the Model without delay}
\label{sec3}

In this section, we analyze the dynamics of the system without delays, i.e.,
system \eqref{mod1}. We derive the basic reproduction number for the system
and the stability of equilibria is discussed using this threshold.

% -------------------------------------------

\subsection{Existence of equilibria}

Model \eqref{mod1} has three steady states, namely, (i) the disease-free equilibrium,
\(E_0\left(\frac{\lambda_{1}}{d_{T}},0,0,\frac{\lambda_{2}}{d_{E}},0\right)\),
(ii) the CTL response--free equilibrium \(E_1(\tilde{T},\tilde{I},\tilde{V},\tilde{E},0)\),
where
\[
\begin{aligned}
\tilde{T}&={\frac {d_{{V}}}{\beta\,m \tilde{E}}},\\
\tilde{I}&={\frac {\beta\,m\lambda_{{1}}\tilde{E}-d_{{T}}d_{{V}}}{\beta\,md_{{I}}\tilde{E}}},\\
\tilde{V}&={\frac {\beta\,m\lambda_{{1}}-d_{{T}}d_{{V}}\tilde{E}}{\beta\,d_{{V}}\tilde{E}}},\\
\tilde{E}&={\frac {-\big(\beta\,m\theta\,\lambda_{{1}}-\beta\,m\lambda_{{2}}\big)
+\sqrt { \big( \beta\,m\theta\,\lambda_{{1}}-\beta\,m\lambda_{{2}}\big)^{2}
+4\,d_{{E}}\beta\,m\theta\,d_{{T}}d_{{V}}}}{2\,\beta\,m\,d_{{E}}}},
\end{aligned}
\]
and (iii) the endemic equilibrium \(E^*\) which is given by
\[
\begin{aligned}
T^*=&{\frac {\lambda_{{1}}\alpha-\alpha\,pC_{{\max}}I^*-d_{{I}}\alpha I^*
+pC_{{\max}}d_{{c}}}{\alpha\,d_{{T}}}},\\
V^*=&{\frac {md_{{I}}I^*}{d_{{V}}}},	\\
E^*=&{\frac {\lambda_{{2}}{\alpha}^{2}I^*-\theta\, \left( \alpha\,I^*
+q \left(\alpha\,I^*-d_{{c}} \right) C_{{\max}} \right)  \left( \alpha\,d_{{I}}
+p\left( \alpha\,I^*-d_{{c}} \right) C_{{\max}} \right) }{{\alpha}^{2}d_{{E}}I^*}},\\
C^*=&{\frac {C_{{\max}} \left( \alpha\,I^*-d_{{c}} \right) }{\alpha\,I^*}},
\end{aligned}
\]
where \(I^*\) is the positive root of the fourth degree equation
\begin{equation}
\label{eq3}
a_0 I^4+a_1 I^3+a_2 I^2+a_3 I+a_4=0
\end{equation}
with
\[
\begin{aligned}
a_0=&{\alpha}^{3}\beta\,m\theta\,d_{{Z}}
\left( pC_{{\max}}+d_{{Z}}\right) ^{2} \left( qC_{{\max}}+1 \right),\\
a_1=&- \left( 3\,pq\theta\,d_{{c}}{C_{{\max}}}^{2}+\theta\, \left( \alpha\,
q\lambda_{{1}}+qd_{{Z}}d_{{c}}+2\,pd_{{c}} \right) C_{{\max}}+\alpha\,
\left( \theta\,\lambda_{{1}}+\lambda_{{2}} \right)  \right) d_{{Z}}m{
\alpha}^{2} \left( pC_{{\max}}+d_{{Z}} \right) \beta,\\
a_2=&2\,d_{{Z}}mq\alpha\,{C_{{\max}}}^{2}pd_{{c}} \left(  \left( 3/2\,
pC_{{\max}}+d_{{Z}} \right) d_{{c}}+\lambda_{{1}}\alpha \right) \theta\,\beta\\
&{}+\left( -C_{{\max}}pq{\alpha}^{2}d_{{E}}d_{{T}}d_{{V}}+md_{{Z}}d_{{c}}
\lambda_{{1}}\beta\,\theta\, \left( qd_{{Z}}+p \right) \alpha+md_{{Z}}
{d_{{c}}}^{2}C_{{\max}}{p}^{2}\beta\,\theta \right) \alpha\,C_{{\max}}\\
&{}+{\alpha}^{2} \left(  \left(  \left( -qC_{{\max}}d_{{E}}d_{{T}}d_{{V}}
+\beta\,m\lambda_{{1}}\lambda_{{2}}-d_{{E}}d_{{T}}d_{{V}} \right)
d_{{Z}}-C_{{\max}}pd_{{E}}d_{{T}}d_{{V}} \right) \alpha+md_{{c}}C_{{\max}}
\lambda_{{2}}p\beta\,d_{{Z}} \right),\\
a_3=&- \left( md_{{Z}}{d_{{c}}}^{2}{C_{{\max}}}^{2}{p}^{2}q\beta\,\theta
+pq \alpha\, \left( \beta\,m\theta\,d_{{Z}}d_{{c}}\lambda_{{1}}-2\,\alpha
\,d_{{E}}d_{{T}}d_{{V}} \right) C_{{\max}}-{\alpha}^{2}d_{{E}}d_{{T}}
d_{{V}} \left( qd_{{Z}}+p \right)  \right) d_{{c}}C_{{\max}},\\
a_4=&-\alpha\,pq{C_{{\max}}}^{2}d_{{E}}d_{{T}}d_{{V}}{d_{{c}}}^{2}.
\end{aligned}
\]

\begin{remark}
We have \(a_0 >0\) and \(a_4 <0\). Thus, equation \eqref{eq3} has at least
one positive real root. For a feasible endemic equilibrium, we also need
\[
\frac{d_c}{\alpha}<I^*
<\min\biggl\{\frac{\lambda_{{1}}\alpha+pd_c C_{\max}}{\alpha(pC_{\max}+d_I)},
\frac{\alpha d_I+qd_c C_{\max}}{\theta \alpha-\theta\alpha q C_{\max}-\lambda^2\alpha^2}\biggr\}.
\]
\end{remark}

% -------------------------------------------

\subsection{Stability of the equilibria}

For the stability of equilibria, we need the nature of the roots of the characteristic
equation of the Jacobian matrix at any equilibrium point $E(T, I, V, E,C)$ \cite{Tekle}.
We linearize the system \eqref{mod1} using the Jacobian technique.
The basic threshold number $R_0$ of system \eqref{mod1} is defined
by the spectral radius of the matrix $FV^{-1}$, and is given by \cite{Abraha}:
\begin{equation}
\label{eq:R0}
R_0={\frac {m\beta\,\lambda_{{1}}\lambda_{{2}}}{d_{{V}}{d_{{E}}}d_{{T}}}}.
\end{equation}
The Jacobian matrix of the system \eqref{mod1} at any equilibrium point \(E(T, I, V, E,C)\)
is given by
\begin{equation}
\label{jacobian}
J(E)=\left[ \begin {array}{ccccc} -{\frac {\beta\,EV}{Cq+1}}-d_{{T}}&0&-{
\frac {\beta\,ET}{Cq+1}}&-{\frac {\beta\,VT}{Cq+1}}&{\frac {\beta\,EVT
q}{ \left( Cq+1 \right) ^{2}}}\\ \noalign{\medskip}{\frac {\beta\,EV}{
Cq+1}}&-Cp-d_{{I}}&{\frac {\beta\,ET}{Cq+1}}&{\frac {\beta\,VT}{Cq+1}}
&-{\frac {\beta\,EVTq}{ \left( Cq+1 \right) ^{2}}}-pZ\\
\noalign{\medskip}0&md_{{I}}&-d_{{V}}&0&0\\ \noalign{\medskip}-
\theta\,\beta\,EV&0&-\beta\,ET\theta&-\beta\,VT\theta-d_{{E}}&0\\
\noalign{\medskip}0&\alpha\,C \left( 1-{\frac {C}{C_{{\max}}}}
\right) &0&0&\alpha\,I \left( 1-\,{\frac 2{C}{C_{{\max}}}} \right)
-d_{{c}}\end {array} \right].
\end{equation}

\begin{theorem}
The disease free equilibrium 
\(E_0\left(\frac{\lambda_{1}}{d_{T}},0,0,\frac{\lambda_{2}}{d_{E}},0\right)\)
is stable for \(R_{0}<1\) and unstable for \(R_{0}>1\).
\end{theorem}

\begin{proof}
From \eqref{jacobian}, the Jacobian matrix at the disease-free equilibrium point
\(E_0\left(\frac{\lambda_{1}}{d_{T}},0,0,\frac{\lambda_{2}}{d_{E}},0\right)\)
is given by
\[
J({E_0})=\left[ \begin {array}{ccccc} -d_{{T}}&0&-{\frac {\beta\,\lambda_{{1}}
\lambda_{{2}}}{d_{{T}}d_{{E}}}}&0&0\\ \noalign{\medskip}0&-d_{{I}}&{
\frac {\beta\,\lambda_{{1}}\lambda_{{2}}}{d_{{T}}d_{{E}}}}&0&0\\
\noalign{\medskip}0&{\it md}_{{I}}&-d_{{V}}&0&0\\
\noalign{\medskip}0&0
&-{\frac {\beta\,\lambda_{{1}}\lambda_{{2}}\theta}{d_{{T}}d_{{E}}}}&-d_{{E}}&0\\
\noalign{\medskip}0&0&0&0&-d_{{c}}\end {array} \right],
\]
whose characteristic equation in $\rho$ is given by
\begin{equation*}
\bigl|\rho I-J(E_0)\bigr|
=\left|\begin{array}{ccccc}
\rho+d_{{T}}&0&{\frac {\beta\,\lambda_{{1}}
\lambda_{{2}}}{d_{{T}}d_{{E}}}}&0&0\\ \noalign{\medskip}0&\rho
+d_{{Z}}&-{\frac {\beta\,\lambda_{{1}}\lambda_{{2}}}{d_{{T}}d_{{E}}}}&0&0\\
\noalign{\medskip}0&-{\it md}_{{Z}}&\rho+d_{{V}}&0&0\\
\noalign{\medskip}0&0&{\frac {\beta\,\lambda_{{1}}\lambda_{{2}}
\theta}{d_{{T}}d_{{E}}}}&\rho+d_{{E}}&0\\ \noalign{\medskip}0&0&0&0&\rho+d_{{c}}
\end {array}
\right|=0.
\end{equation*}
From this we obtain that
\begin{equation}
\left( \rho+d_{{T}} \right)  \left( \rho+d_{{E}} \right)  \left( \rho
+d_{{c}} \right) \left[{\rho}^{2}+ \left( d_{{V}}+d_{{I}} \right) \rho
+d_{{E}}d_{{T}}d_{{V}}d_{{I}}-\beta\,\lambda_{{1}}
\lambda_{{2}}{\it md}_{{I}}\right]=0.
\end{equation}
Hence, using the Routh--Hurwitz criterion, we conclude that
\(E_{0}\) is stable if and only if
\[
d_{{V}} {d_{{E}}}
d_{{T}}  >m\beta\,\lambda_{{1}}\lambda_{{2}},
\]
which is equivalent to \(R_0<1\).
\end{proof}

\begin{theorem}
The CTL-free equilibrium, \(E_1(\tilde{T},\tilde{I},\tilde{V},\tilde{E},0)\),
is asymptotically stable if and only if the following conditions
are satisfied:
\[
\begin{aligned}
&\alpha_1>0,\\
&\alpha_2>0,\\
&\alpha_3>0,\\
&\alpha_4>0,\\
&\alpha_1\,\alpha_2-\alpha_3>0,\\
&(\alpha_1\,\alpha_2-\alpha_3)\alpha_3-\alpha_1^2\alpha_4>0,
\end{aligned}
\]
where
\begin{equation}
\label{alphas:i:1:4}
\begin{split}
\alpha_{{1}}=&-a_{{33}}-a_{{44}}-a_{{11}}-a_{{22}},\\
\alpha_{{2}}=& \left( a_{{33}}+a_{{44}}+a_{{22}} \right) a_{{11}}
-a_{{23}}a_{{32}}+ \left( a_{{33}}+a_{{22}} \right) a_{{44}}
+a_{{33}}a_{{22}}-a_{{41}}a_{{14}},\\
\alpha_{{3}}=& \left( a_{{23}}a_{{32}}+ \left( -a_{{33}}
-a_{{22}}\right) a_{{44}}-a_{{33}}a_{{22}} \right) a_{{11}}\\
&{}+ \left( -a_{{13}}a_{{21}}+a_{{23}}a_{{44}}
-a_{{24}}a_{{43}} \right) a_{{32}}-a_{{33}}a_{{22}}a_{{44}}
+a_{{41}}a_{{14}} \left( a_{{33}}+a_{{22}} \right),\\
\alpha_{{4}}= &\left(  \left( -a_{{23}}a_{{44}}
+a_{{24}}a_{{43}}\right) a_{{32}}+a_{{33}}a_{{22}}a_{{44}} \right) a_{{11}}\\
&{}+ \left(a_{{21}}a_{{13}}a_{{44}}+ \left( -a_{{21}}a_{{43}}
+a_{{23}}a_{{41}}\right) a_{{14}}-a_{{41}}a_{{13}}a_{{24}} \right)
a_{{32}}-a_{{33}}a_{{41}}a_{{14}}a_{{22}},
\end{split}
\end{equation}
with
\begin{equation}
\label{eq:ais}
\begin{split}
a_{{11}}&=-\beta\,EV-d_{{T}}, \quad a_{{21}}=\beta\,EV, \quad a_{{41}}=-\theta\,\beta\,EV, \\
a_{22}&=-d_{I}, \quad a_{32}=md_{I},\\
a_{13}&=-\beta E T, \quad a_{23}=\beta ET, \quad a_{33}=-d_{V}, \quad a_{43}=-\beta\theta ET,\\
a_{14}&=-\beta E T, \quad a_{24}=\beta ET, \quad a_{44}=-d_{V}, \quad a_{43}=-\beta\theta ET,\\
a_{15}&=\beta\,EVTq, \quad a_{25}=-\beta\,EVTq-pI, \quad a_{55}=-d_{c}.
\end{split}
\end{equation}
\end{theorem}

\begin{proof}
Using \eqref{jacobian}, the Jacobian matrix of \eqref{mod1} at
\(E_1(\tilde{T},\tilde{I},\tilde{V},\tilde{E},0)\) is given by
\[
J({E_1})=\left[ \begin {array}{ccccc} a_{{11}}&0&a_{{13}}&a_{{14}}&a_{{15}}\\
\noalign{\medskip}a_{{21}}&a_{{22}}&a_{{23}}&a_{{24}}&a_{{25}}\\
\noalign{\medskip}0&a_{{32}}&a_{{33}}&0&0\\
\noalign{\medskip}a_{{41}}&0&a_{{43}}&a_{{44}}&0\\
\noalign{\medskip}0&0&0&0&a_{{55}}
\end{array} \right],
\]
with the $a_{i j}$ given by \eqref{eq:ais},
whose characteristic equation in $\rho$ is
\begin{equation}
\left( \rho-a_{{55}} \right)  \left( {\rho}^{4}+\alpha_{{1}}{\rho}^{3}
+\alpha_{{2}}{\rho}^{2}+\alpha_{{3}}\rho+\alpha_{{4}} \right) =0,
\end{equation}
with the $\alpha_{{i}}$, $i = 1, \ldots, 4$, given by \eqref{alphas:i:1:4}.
Here, one eigenvalue is \(-d_{c}\) and the rest
of the eigenvalues satisfy the following equation:
\[
{\rho}^{4}+\alpha_{{1}}{\rho}^{3}
+\alpha_{{2}}{\rho}^{2}+\alpha_{{3}}\rho+\alpha_{{4}}  =0.
\]
Hence, using the Routh--Hurwitz criterion, we arrive to the intended result.
\end{proof}

\begin{theorem}
The coexisting equilibrium \(E^*\) is asymptotically
stable if and only if the following conditions are satisfied:
\begin{equation}
\label{eq:cond:B:as:E*}
\begin{gathered}
B_{5}>0,\\
B_{1}B_{2}-B_{3}>0,\\
B_{3}(B_{1}B_{2}-B_{3})-B_{1}(B_{1}B_{4}-B_{5})>0,\\
(B_{1}B_{2}-B_{3})(B_{3}B_{4}-B_{2}B_{5})-(B_{1}B_{4}-B_{5})^2>0,
\end{gathered}
\end{equation}
where
\begin{equation}
\label{Bis:1:5}
\begin{split}
B_{{1}}&=-b_{{33}}-b_{{44}}-b_{{55}}-b_{{11}}-b_{{22}},\\
B_{{2}}&= \left( b_{{33}}+b_{{44}}+b_{{11}}+b_{{22}} \right) b_{{55}}
+ \left( b_{{44}}+b_{{11}}+b_{{22}} \right) b_{{33}}+ \left( b_{{11}}+b_{{22}}
\right) b_{{44}}\\
&\quad -b_{{22}}b_{{11}}-b_{{14}}b_{{41}}-b_{{23}}b_{{32}}+b_{{25}}b_{{52}},\\
B_{{3}}&= \left(  \left( -b_{{44}}-b_{{11}}-b_{{22}} \right) b_{{33}}+ \left( -
b_{{11}}-b_{{22}} \right) b_{{44}}+b_{{23}}b_{{32}}+b_{{14}}b_{{41}}
-b_{{22}}b_{{11}} \right) b_{{55}},\\
&\quad +\left(  \left( -b_{{11}}-b_{{22}} \right) b_{{44}}+b_{{25}}b_{{52}}
+b_{{14}}b_{{41}}-b_{{22}}b_{{11}} \right) b_{{33}}
+ \left( -b_{{22}}b_{{11}}+b_{{23}}b_{{32}}+b_{{25}}b_{{52}} \right) b_{{44}}\\
&\quad +\left( b_{{23}}b_{{32}}+b_{{25}}b_{{52}} \right) b_{{11}}+ \left(
-b_{{13}}b_{{21}}-b_{{24}}b_{{43}} \right) b_{{32}}+b_{{41}}b_{{14}}
b_{{22}}-b_{{52}}b_{{15}}b_{{21}},\\
B_{{4}}&= \biggl(  \left(  \left( b_{{11}}+b_{{22}} \right) b_{{44}}
-b_{{14}}b_{{41}}+b_{{22}}b_{{11}} \right) b_{{33}}\\
&\quad + \left( b_{{22}}b_{{11}}-b_{{23}}b_{{32}} \right)
b_{{44}}-b_{{23}}b_{{32}}b_{{11}}+ \left( b_{{13}}
b_{{21}}+b_{{24}}b_{{43}} \right) b_{{32}}-b_{{41}}b_{{14}}b_{{22}}\biggr) b_{{55}}\\
&\quad +\left(  \left( b_{{22}}b_{{11}}-b_{{25}}b_{{52}} \right) b_{{44}}
-b_{{25}}b_{{52}}b_{{11}}-b_{{41}}b_{{14}}b_{{22}}+b_{{52}}b_{{15}}b_{{21}} \right) b_{{33}}\\
&\quad + \left(  \left( -b_{{23}}b_{{32}}-b_{{25}}b_{{52}}
\right) b_{{11}}+b_{{21}} \left( b_{{13}}b_{{32}}+b_{{15}}b_{{52}}\right)\right) b_{{44}},\\
B_{5}&=	\biggl( {b_{{22}} \left( -b_{{11}}b_{{44}}+b_{{14}}b_{{41}} \right)
b_{{33}}+b_{{32}} \left(  \left( b_{{11}}b_{{23}}-b_{{13}}b_{{21}}\right)
b_{{44}}-b_{{24}}b_{{43}}b_{{11}}+ \left( b_{{13}}b_{{24}}
-b_{{14}}b_{{23}} \right) b_{{41}} \right)}\biggr) b_{{55}}\\
&\quad +\left(b_{{43}}b_{{14}}b_{{21}}\right)b_{{55}}
-b_{{52}} \left(  \left( -b_{{11}}b_{{25}}+b_{{15}}b_{{21}} \right)
b_{{44}}+b_{{41}} \left( b_{{14}}b_{{25}}-b_{{15}}b_{{24}} \right)\right) b_{{33}},
\end{split}
\end{equation}
with the coefficients $b_{{i j}}$ given by
\begin{equation}
\label{eq:bis}
\begin{aligned}
b_{{11}}=&-{\frac {EV\beta}{Cq+1}}-d_{{T}},
\quad b_{{21}}={\frac {EV\beta}{Cq+1}},
\quad b_{{41}}=-\theta\,\beta\,EV, \\
b_{22}=&-Cp-d_{{Z}},
\quad b_{32}=md_{I},
\quad b_{52}=\alpha\,C \left( 1-{\frac {C}{C_{{\max}}}} \right),\\
b_{13}=&-{\frac {\beta\,ET}{Cq+1}},
\quad b_{23}={\frac {\beta\,ET}{Cq+1}},
\quad b_{33}=-d_{V},
\quad b_{43}=-\beta\theta ET,\\
b_{14}=&-{\frac {\beta\,VT}{Cq+1}},
\quad b_{24}={\frac {\beta\,VT}{Cq+1}},
\quad b_{44}=-TV\beta\,\theta-d_{{E}}, \\
b_{15}=&{\frac {\beta\,EVTq}{ \left( Cq+1 \right) ^{2}}},
\quad b_{25} =-{\frac {\beta\,EVTq}{ \left( Cq+1 \right) ^{2}}}-pZ,
\quad b_{55} =-{\frac {\alpha\,Z \left( -C_{{\max}}+2\,C \right) }{C_{{\max}}}}-d_{c}.
\end{aligned}
\end{equation}
\end{theorem}

\begin{proof}
Using \eqref{jacobian}, the Jacobian matrix of \eqref{mod1}
at the coexistence equilibrium point \(E^*\) is given as
\[
J({E^*})=\left[ \begin {array}{ccccc} b_{{11}}&0&b_{{13}}&b_{{14}}&b_{{15}}\\
\noalign{\medskip}b_{{21}}&b_{{22}}&b_{{23}}&b_{{24}}&b_{{25}}\\
\noalign{\medskip}0&b_{{32}}&b_{{33}}&0&0\\
\noalign{\medskip}b_{{41}}&0&b_{{43}}&b_{{44}}&0\\
\noalign{\medskip}0&b_{{52}}&0&0&b_{{55}}
\end{array}
\right],
\]
with the $b_{i j}$ given in \eqref{eq:bis},
whose characteristic equation in $\rho$ is given by
\begin{equation}
\label{char}
{\rho}^{5}+ B_{1}{\rho}^{4}+B_{{2}}{\rho}^{3}
+B_{{3}}{\rho}^{2}+B_{{4}}\rho+B_{{5}} =0
\end{equation}
with $B_{{i}}$, $i = 1, \ldots, 5$, given by \eqref{Bis:1:5}.
Hence, using the Routh--Hurwitz criterion, all roots of the characteristic
equation \eqref{char} have negative real parts,
provided the conditions \eqref{eq:cond:B:as:E*} hold.
\end{proof}

% -------------------------------------------

\section{Analysis of the model with delays}
\label{sec4}

Let \(\psi\) denote the  Banach space of continuous functions
\(\psi:[-\tau,0]\rightarrow\mathbb R^{5}\) equipped with the sup-norm
\[
\parallel\psi\parallel
=\sup_{-\tau\leq\theta\leq0}\left\{\mid\psi_{1}(\theta)\mid,\mid\psi_{2}(\theta)\mid,
\mid\psi_{3}(\theta)\mid,\mid\psi_{4}(\theta)\mid,\mid\psi_{5}(\theta)\mid\right\}.
\]
For biological reasons, populations always have non-negative
values. Therefore, the initial functions for our model \eqref{mod2}
are given as below \cite{Abraha}:
\begin{equation}
\begin{cases}
\label{eq:ic}
T(\theta)=\psi_{1}(\theta), \quad I(\theta)
=\psi_{2}(\theta), \quad V(\theta)
=\psi_{3}(\theta), \quad E(\theta)=\psi_{4}(\theta), \quad C(\theta)=\psi_{5}(\theta),\\
\psi_{i}(\theta)\geq0,\quad \theta\in[-\tau,0],\quad \psi_{i}(0)>0,
\quad \psi=(\psi_{1},\psi_{2},\psi_{3},\psi_{4},\psi_{5}) \in \mathcal{C}([-\tau,0],\mathbb{R}^{5}_{+}),	
\end{cases}
\end{equation}
where $\tau=\max\{\tau_1,\tau_2\}$ and
\(
\begin{aligned}
\mathbb {R}_{+}^{5}
=\left\{ \left( x_{1},x_{2},x_{3},x_{4},x_{5}\right)
\mid x_{i}\ge 0, i=1,2,3,4,5\right\}.
\end{aligned}
\)
The system \eqref{mod2} exhibits three equilibrium points:
(i) the disease-free equilibrium,
\(E_0\left(\frac{\lambda_{1}}{d_{T}},0,0,\frac{\lambda_{2}}{d_{E}},0\right)\),
(ii) the CTL response--free equilibrium \(E_1(\tilde{T},\tilde{I},\tilde{V},\tilde{E},0)\),
where
\[
\begin{aligned}
\tilde{T}&={\frac {d_{{V}}}{E\beta\,m \left( \epsilon _{{1}}\epsilon _{{2}}
-\epsilon _{{1}}-\epsilon _{{2}}+1 \right) }},\\
\tilde{I}&={\frac {Em\beta\, \left( \epsilon _{{2}}-1 \right)
\left( \epsilon_{{1}}-1 \right) \lambda_{{1}} \left(
1-\epsilon _{{1}} \right) -d_{{T}}d_{{V}}}{E\beta\,md_{{Z}}
\left( \epsilon _{{2}}-1 \right) }},\\
\tilde{V}&={\frac {Em\beta\, \left( \epsilon_{{2}}
-1 \right)  \left( \epsilon _{{1}}-1 \right)
\lambda_{{1}} \left( 1-\epsilon _{{1}} \right)-d_{{T}}
d_{{V}}}{E\beta\,d_{{V}}}},
\end{aligned}
\]
and \(\tilde{E}\) is the positive solution of
\(
\lambda_{{2}}-{\frac {\theta\, \left( Em\beta\, \left( \epsilon _{{2}}
-1 \right)  \left( \epsilon _{{1}}-1 \right) \lambda_{{1}} \left(
1-\epsilon _{{1}} \right) -d_{{T}}d_{{V}} \right) }{Em\beta\,
\left(\epsilon _{{2}}\epsilon _{{1}}-\epsilon _{{1}}-\epsilon _{{2}}
+1\right) }}-d_{{E}}E =0,
\)
and (iii) the endemic equilibrium \(E^*\), which
is the positive solution of the nonlinear system
\[
\begin{cases}
\begin{aligned}
&\lambda_1 - (1-\epsilon_1)\frac{\beta E V T}{1+q C}  - d_T T=0,\\
&(1-\epsilon_1)\frac{e^{-d_I\tau_1}\beta E(t-\tau_1) V(t-\tau_1)
T(t-\tau_1)}{1+q C(t-\tau_1)} - d_I I - p I C=0,\\
& (1-\epsilon_2)m d_I I - d_V V=0,\\
& \lambda_2 - \theta \beta E V T - d_E E=0,\\
&\alpha I(t-\tau_2) C(t-\tau_2) \left(1 - \frac{C(t-\tau_2)}{C_{\max}}\right)
- d_c C=0.
\end{aligned}
\end{cases}
\]
If we linearize  our system \eqref{mod2} about
a steady state $E(T,I,V,E,C)$, we obtain that
\begin{equation}
\label{2}
\frac{dX}{dt}=F X(t)+GX(t-\tau_1)+HX(t-\tau_2).
\end{equation}
Here, $F$, $G$ and $H$ are $5\times5$ matrices given as
\[
\begin{aligned}
&F=[F_{ij}]=\left[ \begin {array}{ccccc} -{\frac {\beta\,EV}{qC+1}}-d_{{T}}&0
&-{\frac {\beta\,ET}{qC+1}}&-{\frac {\beta\,VT}{qC+1}}
&{\frac {\beta\,EVTq}{ \left( qC+1 \right) ^{2}}}\\
\noalign{\medskip}{\frac { \left( 1-\epsilon_{{1}} \right)
\beta\,EV}{qC+1}}&-pC-d_{{I}}&{\frac { \left(
1-\epsilon _{{1}} \right) \beta\,ET}{qC+1}}
&{\frac { \left( 1-\epsilon _{{1}} \right) \beta\,VT}{qC+1}}
&-{\frac { \left( 1-\epsilon_{{1}} \right) \beta\,qEVT}{
\left( qC+1 \right) ^{2}}}-pI\\
\noalign{\medskip}0& \left( 1-\epsilon _{{2}} \right) md_{{I}}&
-d_{{V}}&0&0\\ \noalign{\medskip}-\theta\,\beta\,EV&0&-\beta\,ET\theta
&-\beta\,VT\theta-d_{{E}}&0\\ \noalign{\medskip}0&0 &0&0&-d_{{c}}
\end{array}
\right],
\end{aligned}
\]
\[
\begin{aligned}
&G=[G_{ij}]=\left[ \begin{array}{ccccc} 0&0&0&0&0\\ \noalign{\medskip}{
\frac{\left(1-\epsilon _{{1}} \right)e^{-d_{I}\tau_{1}} \beta\,EV}{1+qC}}
&0&{\frac{\left(1-\epsilon _{{1}} \right)e^{-d_{I}\tau_{1}} \beta\,ET}{1+qC}}
&{\frac{\left(1-\epsilon _{{1}} \right)e^{-d_{I}\tau_{1}} \beta\,VT}{1+qC}}
&-{\frac { \left( 1-\epsilon_{{1}} \right)e^{-d_{I}\tau_{1}} \beta\,EVT}{
\left( qC+1 \right) ^{2}}}\\
\noalign{\medskip}0& 0&0&0&0\\
\noalign{\medskip}0&0&0&0&0\\
\noalign{\medskip}0&0 &0&0&0
\end{array}\right],\\
&H=[H_{ij}]=\left[
\begin{array}{ccccc}
0 & 0 & 0 & 0&0\\
0 & 0 & 0 & 0&0 \\
0 & 0 & 0 & 0&0 \\
0 & 0 & 0 & 0&0 \\
0 & \alpha C(1-\frac{C}{C_{\max}}) & 0 & 0&\alpha I(1-\frac{2C}{C_{\max}})
\end{array}\right].
\end{aligned}
\]
The characteristic equation of the delayed system \eqref{mod2} then becomes
\begin{equation*}
\triangle(\xi,\tau_1,\tau_2) =\mid {\xi}{I}-F-e^{-{\xi}{\tau_1}}G-e^{-{\xi}{\tau_2}}H\mid=0.
\end{equation*}
This yields the characteristic equation
\begin{multline}
\label{Char_eqn2}
0 = \xi^5+L_1\xi^4+L_2\xi^3+L_3\xi^2+L_4\xi+L_5
+e^{-\xi\tau_1}(\xi^4+M_2\xi^3+M_3\xi^2+M_4\xi+M_5)\\
+e^{-\xi\tau_2}(\xi^5+K_1\xi^4+K_2\xi^3+K_3\xi^2+K_4\xi+K_5)
+e^{-\xi(\tau_1+\tau_2)}(\xi^4+Q_2\xi^3+Q_3\xi^2+Q_4\xi+Q_5),
\end{multline}
where the coefficients are given by
\begin{eqnarray*}
L_1&=& - F_{11} - F_{22} - F_{33} - F_{44} - F_{55},\\
L_2&=&+ F_{11} F_{22} - F_{23} F_{32} + F_{11} F_{33} + F_{22} F_{33} - F_{14} F_{41} + F_{11} F_{44}\\
&&+F_{22} F_{44} + F_{33} F_{44} + F_{11} F_{55} + F_{22} F_{55} + F_{33} F_{55} + F_{44} F_{55},\\
L_3&=&- F_{13} F_{21} F_{32} + F_{11} F_{23} F_{32} - F_{11} F_{22} F_{33} + F_{14} F_{22} F_{41}\\
&&+F_{14} F_{33} F_{41} - F_{24} F_{32} F_{43} - F_{11} F_{22} F_{44}\\
&&+F_{23} F_{32} F_{44} - F_{11} F_{33} F_{44} - F_{22} F_{33} F_{44}\\
&&-F_{11} F_{22} F_{55} + F_{23} F_{32} F_{55} - F_{11} F_{33} F_{55}\\
&& -F_{22} F_{33} F_{55} + F_{14} F_{41} F_{55} - F_{11} F_{44} F_{55}\\
&&-F_{22} F_{44} F_{55} - F_{33} F_{44} F_{55},\\
L_4&=&+F_{14} F_{23} F_{32} F_{41} - F_{13} F_{24} F_{32} F_{41} - F_{14} F_{22} F_{33} F_{41}\\
&& -F_{14} F_{21} F_{32} F_{43} + F_{11} F_{24} F_{32} F_{43} + F_{13} F_{21} F_{32} F_{44}\\
&& -F_{11} F_{23} F_{32} F_{44} + F_{11} F_{22} F_{33} F_{44} + F_{13} F_{21} F_{32} F_{55}\\
&& -F_{11} F_{23} F_{32} F_{55} + F_{11} F_{22} F_{33} F_{55} - F_{14} F_{22} F_{41} F_{55}\\
&& -F_{14} F_{33} F_{41} F_{55} + F_{24} F_{32} F_{43} F_{55} + F_{11} F_{22} F_{44} F_{55}\\
&& -F_{23} F_{32} F_{44} F_{55} + F_{11} F_{33} F_{44} F_{55} + F_{22} F_{33} F_{44} F_{55},\\
L_5&=&-F_{14} F_{23} F_{32} F_{41} F_{55} + F_{13} F_{24} F_{32} F_{41} F_{55} + F_{14} F_{22} F_{33} F_{41} F_{55}\\
&& +F_{14} F_{21} F_{32} F_{43} F_{55} - F_{11} F_{24} F_{32} F_{43} F_{55} - F_{13} F_{21} F_{32} F_{44} F_{55}\\
&&+F_{11} F_{23} F_{32} F_{44} F_{55} - F_{11} F_{22} F_{33} F_{44} F_{55},\\
M_2&=&- F_{32} G_{23} ,\\
M_3&=&- F_{13} F_{32} G_{21} + F_{11} F_{32} G_{23} + F_{32} F_{44} G_{23} + F_{32} F_{55} G_{23} - F_{32} F_{43} G_{24},\\
M_4&=&-F_{14} F_{32} F_{43} G_{21} + F_{13} F_{32} F_{44} G_{21} + F_{13} F_{32} F_{55} G_{21}\\
&& +F_{14} F_{32} F_{41} G_{23} - F_{11} F_{32} F_{44} G_{23} - F_{11} F_{32} F_{55} G_{23}\\
&& -F_{32} F_{44} F_{55} G_{23} - F_{13} F_{32} F_{41} G_{24} + F_{11} F_{32} F_{43} G_{24} + F_{32} F_{43} F_{55} G_{24},\\
M_5&=& + F_{14} F_{32} F_{43} F_{55} G_{21} - F_{13} F_{32} F_{44} F_{55} G_{21} - F_{14} F_{32} F_{41} F_{55} G_{23}\\
&&+ F_{11} F_{32} F_{44} F_{55} G_{23} + F_{13} F_{32} F_{41} F_{55} G_{24} - F_{11} F_{32} F_{43} F_{55} G_{24},
\end{eqnarray*}
\begin{eqnarray*}
K_1&=& - H_{55},\\
K_2&=&- F_{25} H_{52}+ F_{11} H_{55} + F_{22} H_{55} + F_{33} H_{55} + F_{44} H_{55} ,\\
K_3&=&- F_{15} F_{21} H_{52} + F_{11} F_{25} H_{52} + F_{25} F_{33} H_{52} +F_{25} F_{44} H_{52}\\
&&-F_{11} F_{22} H_{55} + F_{23} F_{32} H_{55} - F_{11} F_{33} H_{55}\\
&& -F_{22} F_{33} H_{55} + F_{14} F_{41} H_{55} - F_{11} F_{44} H_{55}\\
&& -F_{22} F_{44} H_{55} - F_{33} F_{44} H_{55} ,\\
K_4&=&+ F_{15} F_{21} F_{33} H_{52} - F_{11} F_{25} F_{33} H_{52}\\
&& -F_{15} F_{24} F_{41} H_{52} + F_{14} F_{25} F_{41} H_{52} + F_{15} F_{21} F_{44} H_{52}\\
&& - F_{11} F_{25} F_{44} H_{52} - F_{25} F_{33} F_{44} H_{52}+F_{13} F_{21} F_{32} H_{55} - F_{11} F_{23} F_{32} H_{55}\\
&& + F_{11} F_{22} F_{33} H_{55}-F_{14} F_{22} F_{41} H_{55} - F_{14} F_{33} F_{41} H_{55}
+ F_{24} F_{32} F_{43} H_{55}\\&&
+ F_{11} F_{22} F_{44} H_{55} - F_{23} F_{32} F_{44} H_{55} + F_{11} F_{33} F_{44} H_{55} + F_{22} F_{33} F_{44} H_{55},\\
K_5&=& F_{15} F_{24} F_{33} F_{41} H_{52} - F_{14} F_{25} F_{33} F_{41} H_{52} - F_{15} F_{21} F_{33} F_{44} H_{52}
+ F_{11} F_{25} F_{33} F_{44} H_{52}\\
&&- F_{14} F_{23} F_{32} F_{41} H_{55} + F_{13} F_{24} F_{32} F_{41} H_{55}\\
&& + F_{14} F_{22} F_{33} F_{41} H_{55} + F_{14} F_{21} F_{32} F_{43} H_{55} - F_{11} F_{24} F_{32} F_{43} H_{55}\\
&& - F_{13} F_{21} F_{32} F_{44} H_{55} + F_{11} F_{23} F_{32} F_{44} H_{55} - F_{11} F_{22} F_{33} F_{44} H_{55},
\end{eqnarray*}
and
\begin{eqnarray*}
Q_2&=&  - G_{25} H_{52},\\
Q_3&=&- F_{15} G_{21} H_{52} + F_{11} G_{25} H_{52} + F_{33} G_{25} H_{52} + F_{44} G_{25} H_{52}
+ F_{32} G_{23} H_{55},\\
Q_4&=&+ F_{15} F_{33} G_{21} H_{52} + F_{15} F_{44} G_{21} H_{52} - F_{15} F_{41} G_{24} H_{52}
- F_{11} F_{33} G_{25} H_{52} \\
&& + F_{14} F_{41} G_{25} H_{52} - F_{11} F_{44} G_{25} H_{52} - F_{33} F_{44} G_{25} H_{52}
+F_{13} F_{32} G_{21} H_{55} \\
&&- F_{11} F_{32} G_{23} H_{55} -F_{32} F_{44} G_{23} H_{55} + F_{32} F_{43} G_{24} H_{55}, \\
Q_5&=& - F_{15} F_{33} F_{44} G_{21} H_{52} + F_{15} F_{33} F_{41} G_{24} H_{52}
+F_{11} F_{33} F_{44} G_{25} H_{52}\\
&& + F_{14} F_{32} F_{43} G_{21} H_{55} - F_{13} F_{32} F_{44} G_{21} H_{55} - F_{14} F_{32} F_{41} G_{23} H_{55}\\
&& + F_{11} F_{32} F_{44} G_{23} H_{55} + F_{13} F_{32} F_{41} G_{24} H_{55} - F_{11} F_{32} F_{43} G_{24} H_{55}
- F_{14} F_{33} F_{41} G_{25} H_{52}.
\end{eqnarray*}

% -------------------------------------------

\subsection{Stability and Hopf bifurcation of the delayed system}

There are two cases to be investigated.

% -------------------------------------------

\subsubsection{Case I: $\tau_1=0$ and $\tau_2>0$}

For $\tau_1=0$ and $\tau_2>0$, the characteristic equation is
\begin{equation}
\label{char_delay3}
\psi(\xi,\tau_2)=\xi^5+m_1\xi^4+m_2\xi^3+m_3\xi^2+m_4\xi
+m_5e^{-\xi\tau_2}(n_1\xi^4+n_2\xi^3+n_3\xi^2+n_4\xi+n_5)=0,
\end{equation}
where
\[
m_1=L_1,~~m_2=L_2+M_2, ~~m_3=L_3+M_3,~~m_4=L_4+M_4,~~m_5=L_5+M_5,
\]
\[
n_1=K_1,~~n_2=K_2+Q_2, ~~n_3=K_3+Q_3,~~n_4=K_4+Q_4,~~n_5=K_5+Q_5.
\]
For Hopf bifurcation to occur, we need to find a purely imaginary root
of \eqref{char_delay3}. Let us assume $\xi = i\omega$, $\omega > 0$,
to be a purely imaginary root of (\ref{char_delay3}). Putting $\xi = i\omega$
in \eqref{char_delay3}, and separating real and imaginary parts, we get
\begin{equation}
\label{eqqn5}
\begin{gathered}
m_1\omega^4-m_3\omega^2+m_5
=-(n_1\omega^4-n_3\omega^2+n_5)\cos\omega\tau_2 -(-n_2\omega^3+n_4\omega)\sin \omega \tau_2,\\
\omega^5 - m_2\omega^3 + m_4\omega=(n_1\omega^4 - n_3\omega^2 + n_5)
\sin\omega\tau_2 - (-n_2\omega^3 + n_4\omega) \cos \omega\tau_2.
\end{gathered}
\end{equation}
Squaring and adding the two equations of (\ref{eqqn5}), we obtain that
\begin{eqnarray}
\label{eq_etai}
\omega^{10} + p_1\omega^8 + p_2\omega^6 + p_3\omega^4 + p_4\omega^2 + p_5 = 0,
\end{eqnarray}
where
\begin{eqnarray*}
p_1&=& m_1^2-2m_2-n_1^2,~~~p_2=m_2^2+2m_4-2m_1m_3-n_2^2+2n_1n_3,\\
p_3&=&m_3^2+2m_1m_5-2m_4m_2-2n_1n_5+2n_4n_2-n_3^2,\\
p_4&=&-2m_3m_5+m_4^2+2n_3n_5-n_4^2,~~~p_5=m_5^2-n_5^2.
\end{eqnarray*}
Letting $\xi^2=l$, the equation (\ref{eq_etai}) becomes
\begin{equation}
\label{eq19}
H(l) = l^{5} + p_1l^{4} + p_2l^{3} + p_3l^{2} + p_4l + p_5 = 0.
\end{equation}
Now, for later use, we define:
\begin{equation}
\label{20}
\begin{split}
a_1&=-\frac{6}{25}p_1^2+\frac{3}{5} p_2, \quad b_1
= -\frac{6}{25}p_1p_2+\frac{2}{5}p_3+\frac{8}{125}p_1^3,\\
c_1&=-\frac{3}{625}p_1^4+\frac{3}{125}p_1^2p_2
-\frac{2}{25}p_1p_3+\frac{1}{5}p_4, ~~~\Delta_0={a_1}^2-4c_1,\\
a_2&=-\frac{1}{3}a_1-4c_1,\quad b_2=-\frac{2}{27}{a_1}^3+\frac{8}{3}a_1c_1-{b_1}^2,\\
\Delta_1 &= a_1-4c_1,\quad d_0= \sqrt[3]{-\frac{b_2}{2}+\sqrt\Delta_1}
+ \sqrt[3]{-\frac{b_2}{2}-\sqrt\Delta_1}+\frac{a_1}{3},\\
\Delta_2&=-d_0-a_1+\frac{2b_1}{\sqrt{d_0-a_1}},
\quad \Delta_3 =-d_0-a_1-\frac{2b_1}{\sqrt{d_0-a_1}}.
\end{split}
\end{equation}
Applying the results on the distribution of roots for an equation
of degree five, as derived in \cite{Ref_delay5}, we get Lemma~\ref{lem:d5}.

\begin{lemma}
\label{lem:d5}
For the polynomial equation (\ref{eq19}), the following results hold:\\

\noindent (a) If $p_5 < 0$, then equation (\ref{eq19}) will have at least one positive root.\\

\noindent (b) Suppose $p_5 \geq 0$ and $b_1=0$.
\begin{itemize}
\item[(i)] If $\Delta_0 < 0$, then equation (\ref{eq19}) has no positive real root.
\item[(ii)] If $\Delta_0 \geq0$ and $a_1\geq 0$ and $c_1>0,$ then (\ref{eq19}) has no positive real root.
\item[(iii)] If (i) and (ii) are not satisfied, then (\ref{eq19}) has a positive real root
if and only if there exists at least one positive $l\in\{ l_1,l_2,l_3,l_4\}$ such that $H(l)\leq0$, where
\[
l_i=\delta_i-\frac{p_1}{5}, ~i=1,2,3,4,
\]
and
\[
\delta_1= \sqrt{\frac{-a_1+\sqrt{\Delta_0}}{2}},
\quad  \delta_2= -\sqrt{\frac{-a_1+\sqrt{\Delta_0}}{2}},
\]
\[
\delta_3= \sqrt{\frac{-a_1-\sqrt{\Delta_0}}{2}},
\quad  \delta_4= -\sqrt{\frac{-a_1-\sqrt{\Delta_0}}{2}}.
\]
\end{itemize}
\noindent (c) Suppose that $p_5 \geq 0$, $b_1\neq 0$ and $d_0>b_1$.
\begin{itemize}
\item[(i)] If $\Delta_2 < 0$ and $\Delta_3 < 0$, then (\ref{eq19}) has no positive real root.
\item[(ii)] If (i) is not satisfied, then (\ref{eq19}) has a positive real root if and only if there
exists at least one positive $l\in\{ l_1,l_2,l_3,l_4\}$ such that $H(l)\leq0$, where
\[
l_i=\delta_i-\frac{p_1}{5}, ~i=1,2,3,4,
\]
\end{itemize}
and
\[
\delta_1= \frac{-\sqrt{d_0-a_1}+\sqrt{\Delta_2}}{2},
\quad  \delta_2=\frac{-\sqrt{d_0-a_1}-\sqrt{\Delta_2}}{2},
\]
\[
\delta_3= \frac{-\sqrt{d_0-a_1}+\sqrt{\Delta_3}}{2},
\quad  \delta_4=\frac{-\sqrt{d_0-a_1}-\sqrt{\Delta_3}}{2}.
\]

\noindent (d) Assume that $\Delta_0 \geq0$ and $b_1\neq0$.
Then, (\ref{eq19}) has a positive real root if and only if
$$
\frac{b_1}{4(a_1-d_0)}+\frac{1}{2}d_0=0,
\quad \bar{l}>0,
\quad H(\bar{l})\leq0,
$$
where $\bar{l}=b_1/2(a_1-d_0)-\frac{1}{5}p_1$.
\end{lemma}

Without loss of generality, we have assumed that equation (\ref{eq19})
has $r$ positive roots with $r\in \{1, 2, 3, 4, 5\}$, denoted by
$l_k, k = 1, 2, \ldots r$. Then, equation (\ref{eq_etai}) has $r$ positive roots,
say $\omega_k =\sqrt{l_k}$, $k = 1, 2, \ldots, r$. Therefore, $i\pm\omega_k$
is a pair of purely imaginary roots of \eqref{char_delay3}. Now,
using equation \eqref{eqqn5}, we derive
\begin{equation}
\begin{split}
\sin(\tau_2\eta)
&=\frac{(\omega^5-m_2\omega^3+m_4\omega)(n_1\omega^4
-n_3\omega^2+ n_5)}{(n_1\omega^4-n_3\omega^2+ n_5)^2
+(n_2\omega^3-n_4\omega)^2}
-\frac{(m_1\omega^4-m_3\omega^2+m_5)(-n_2\omega^3+n_4\omega)}{(n_1\omega^4-n_3\omega^2+n_5)^2
+ (n_2\omega^3-n_4\omega)^2}\\
&=\gamma(\eta).
\end{split}
\end{equation}

If equation (\ref{eq_etai}) has at least one positive root,
say $\omega_0$, then (\ref{char_delay3}) will have a pair
of purely imaginary roots $\pm i\omega_0$ corresponding
to the delay $\tau_2^*$. Without loss of generality,
we assume that (\ref{eq_etai}) has ten positive real roots,
which are denoted as $\eta_1$, $\eta_2$, \ldots, $\eta_{10}$,
respectively. For every fixed $\omega_k$, $k = 1, 2, \ldots , 10$,
the corresponding critical values of time delay $\tau_2$ are
\begin{equation*}
\tau_{2,n}^{(j)}
=\frac{1}{\omega_k}\arcsin\gamma(\eta)+\frac{2\pi j}{\omega_k},
\quad j = 0, 1, 2, \ldots
\end{equation*}
Let
\begin{equation}
\label{tau_0}
\tau_2^* = \min\left\{\tau_{2,k}^{(j)}\right\},
\quad
\omega_0 = \omega_k|_{\tau_2=\tau_2^*},
\quad k = 1, 2, \ldots , r,
\quad j=0,1,2, \ldots
\end{equation}
Taking the derivative of \eqref{char_delay3} with respect to $\tau_2$,
it is easy to obtain that
\begin{equation}
\left(\frac{d\xi}{d\tau_2}\right)^{-1}
=-\frac{5\xi^4 + 4m_1\xi^3 + 3m_2\xi^2 + 2m_3\xi
+m_4}{\xi(\xi^5 + m_1\xi^4 + m_2\xi^3+  m_3\xi^2+m_4\xi+m_5)}
+\frac{ 4n_1\xi^3+3n_2\xi^2 + 2n_3\xi
+ n_4}{\xi( n_1\xi^4 + n_2\xi^3 + n_3\xi^2 + n_4\xi+n_5)}
-\frac{\tau_2}{\xi}.
\end{equation}
Thus, we have
\begin{equation}	
\mbox{sign}\left[\frac{d~Re \{\xi(\tau_2)\}}{d\tau_2}\right]_{\tau_2
=\tau_{2,k}^{ (j)}}=\mbox{sign}\left[Re \left(
\frac{d\xi}{d\tau_2}\right)\right]_{\xi=i\omega_0}.
\end{equation}
Using \eqref{eq_etai}, and after simple calculations,
it can be shown that for $\xi=i\omega_0$ one has
\begin{equation}	
\mbox{sign}\left[\frac{d~Re
\{\xi(\tau_2)\}}{d\tau_2}\right]_{\tau_2=\tau_{2,k}^{ (j)} }
=\mbox{sign}\left[\frac{H'(l_k)}{\omega_0^2(n_2\omega_0^2-n_4)^2
+(n_1\omega_0^4-n_3\omega_0^2+n_5)^2}\right].
\end{equation}
We conclude that
$\frac{dRe(\xi(\tau_2))}{d\tau_2}|_{\tau_2=\tau_{2,k}^{ (j)}}$
and $H'(l_k)$ have the same sign.

We summarize the above discussion in the following theorem.

\begin{theorem}
\label{theorem5}
Let $\omega_0$ and $\tau_2^*$ be defined by \eqref{tau_0}.
Consider the conditions
\begin{itemize}
	
\item[(i)] $p_5 < 0$;

\item[(ii)] $p_5\geq0$, $b_1^*=0$, $\Delta_0\geq0$ and $a_1^*<0$
or $c_1^*\leq0$ and there exists $l^*\in\{ l_1,l_2,l_3,l_4\}$
such that $l^*>0$ and $H(l^*)\leq0$;

\item[(iii)] $p_5\geq0$, $b_1^*\neq0$, $d_0^*>b_1^*$
and $\Delta_2 \geq 0$ or $\Delta_3 \geq 0$ and
there exist $l^*\in\{ l_1,l_2,l_3,l_4\}$
such that $l^*>0$ and $H(l^*)\leq0$;

\item[(iv)] if $\Delta_0 \geq0$  and $b_1^*\neq0$,
then (\ref{eq19}) has a positive real root
if and only if  $b_1^*/4(a_1^*-d_0^*)+\frac{1}{2}d_0^*=0$
and $\bar{l}>0$ and $H(\bar{l})\leq0$,
where $\bar{l}=b_1^*/2(a_1^*-d_0^*)-\frac{1}{5}p_1$.
\end{itemize}
The following holds:

\begin{enumerate}
\item If none of conditions (i), (ii), (iii) and (iv) are satisfied,
then the equilibrium $E^*$ is locally asymptotically stable for any $\tau_2\geq0$.

\item If any of the conditions (i), (ii), (iii) or (iv) is satisfied,
then the equilibrium $E^*$ is locally asymptotically
stable for $\tau_2\in[0, \tau_2^*)$.

\item If one of the conditions (i), (ii), (iii) or (iv) holds
and $H(\bar{l_0})\neq0$, then $E^*$ undergoes Hopf bifurcation
as $\tau_2$ passes through the critical value $\tau_2^*$.
\end{enumerate}
\end{theorem}

% -------------------------------------------

\subsubsection{Case II: $\tau_1>0$ and $\tau_2=0$}

At the endemic equilibrium $E^*$, the characteristic equation
has delay-dependent coefficients (i.e., for $\tau_1>0$)
and it is quite involved. Therefore, it is difficult to
obtain analytically information on the nature of the eigenvalues
and on the conditions for occurrence of stability switches.
However, the nature of the eigenvalues can be investigated
at the endemic state through numerical simulations, which
we do in Section~\ref{sec5}.

% -------------------------------------------

\section{Numerical Simulations}
\label{sec5}

In this section, we provide some numerical simulations for illustrating
the dynamics of the system. The numerical simulation of model \eqref{mod2}
is plotted with the basic model parameters as in Table~\ref{table1}.
In Figure~\ref{fig2}, we plot the solutions of the model variables
corresponding to uninfected epithelial cells ($T$), infected epithelial cells ($I$),
SARS-CoV-2 virus ($V$), ACE2 receptor of epithelial cells ($E$),
and CTL responses ($C$) for the non-delayed and delayed system.
The blue line indicates the non-delayed system whereas the red line
represents the delayed system with $\tau_1=5$, keeping $\tau_2=0$.
The initial point is obtained by perturbing $T$ from the non-trivial equilibrium
values $E_1$ given by $(\tilde{T}, \tilde{I}, \tilde{V}, \tilde{E}, \tilde{C})
=(46.78, 2.91,14.47,9.8, 1.05)$, using the set of parameters from Table~\ref{table1}.
In Figure~\ref{fig2}, we choose the initial conditions as $T(\theta)=45$,
$I(\theta)=4$, $V(\theta)=20$, $E(\theta)= 9.5$, $C(\theta)=2$ for $\theta \in (-\tau,0]$.
When $\tau_1=\tau_2=0$ days, the non-trivial equilibrium $E$ is locally asymptotically stable.
But in presence of a delay, the system attains its lower concentration with respect to $p$.
The system shows similar behaviors with respect to the parameter $q$ (see Figure~\ref{fig3}).
Thus, the lytic and nonlytic effect on the model in presence of delay shows almost similar behaviors.

Figures~\ref{fig4} and \ref{fig6} show that increasing the value of the delay
makes the oscillation in the system solution and it persists
for a longer period of time. The parameter values used in $H(l)$
have one positive simple real root and the corresponding value of $\tau_2$ is days.
Our simulations are consistent with the theoretical findings.
When $\tau_1=\tau_2=0$, the endemic equilibrium $E^*$ is locally asymptotically stable.
The initial oscillation persist for a longer period for an increasing value
of the delay $\tau_2$ and the value of the delay passes through its critical value
(see Figures~\ref{fig5} and \ref{fig7}), then the system attains Hopf bifurcation.

\begin{figure}[htb!]
\centering
\includegraphics[width=5in,height=5in]{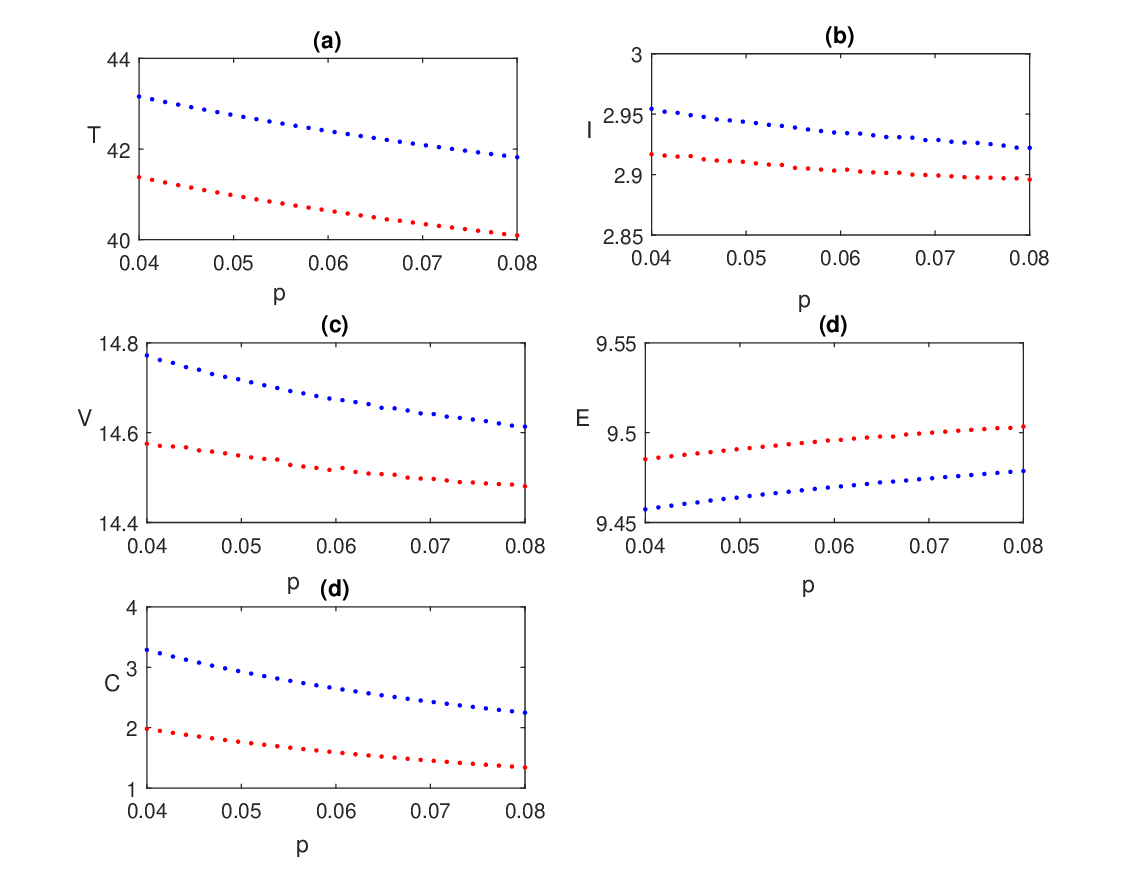}
\caption{Effect of parameter $p$ (the efficacy of the lytic effect)
on the system when $\tau_1=0=\tau_2$.
The other parameter values are given in Table~\ref{table1}.}
\label{fig2}
\end{figure}
\begin{figure}[htb!]
\centering
\includegraphics[width=5in,height=5in]{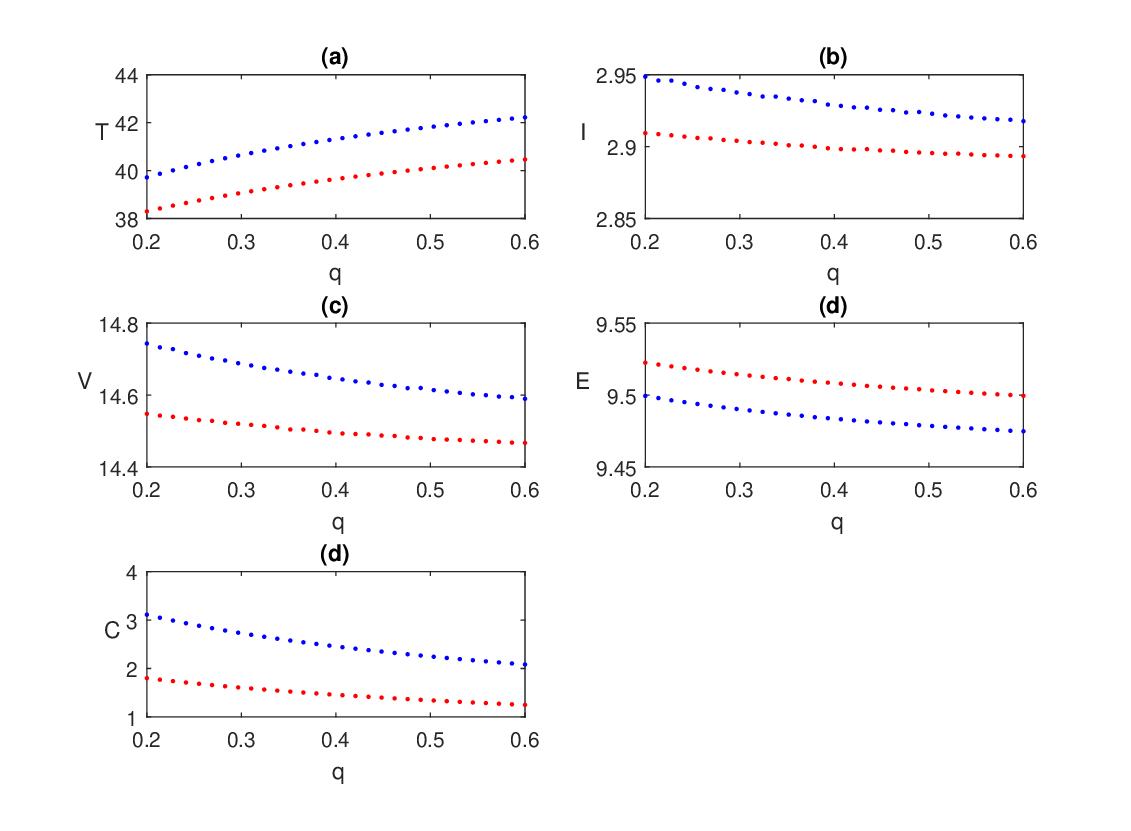}
\caption{Effect of parameter $q$ (the efficacy of the nonlytic effect)
on the system when $\tau_1=0=\tau_2$.
The other parameter values are given in Table~\ref{table1}.}
\label{fig3}
\end{figure}
	
\begin{figure}[htb!]
\centering
\includegraphics[width=5in,height=5in]{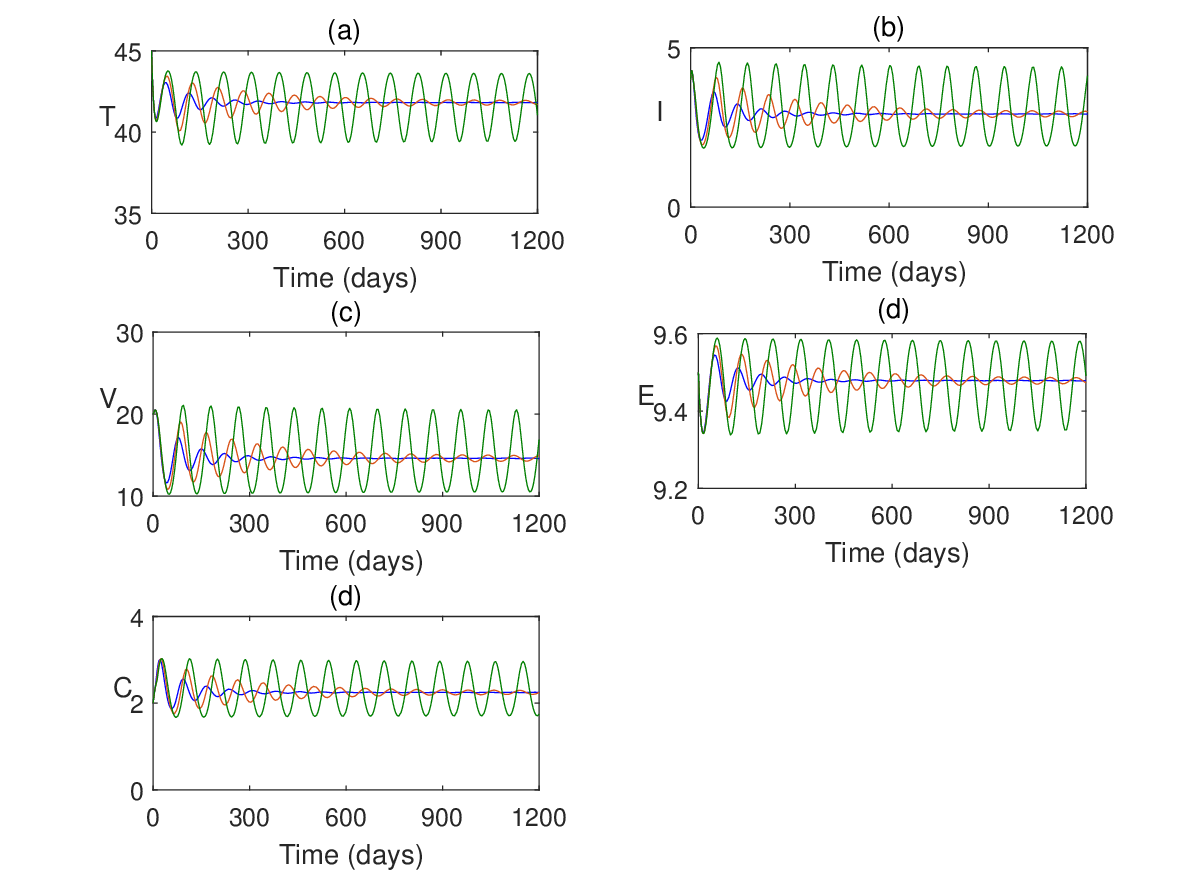}
\caption{Numerical solution of the system with $\tau_1=0$
and for different values of $\tau_2$.
The other parameter values are given in Table~\ref{table1}.}
\label{fig4}
\end{figure}
\begin{figure}[htb!]
\centering
\includegraphics[width=5in,height=5in]{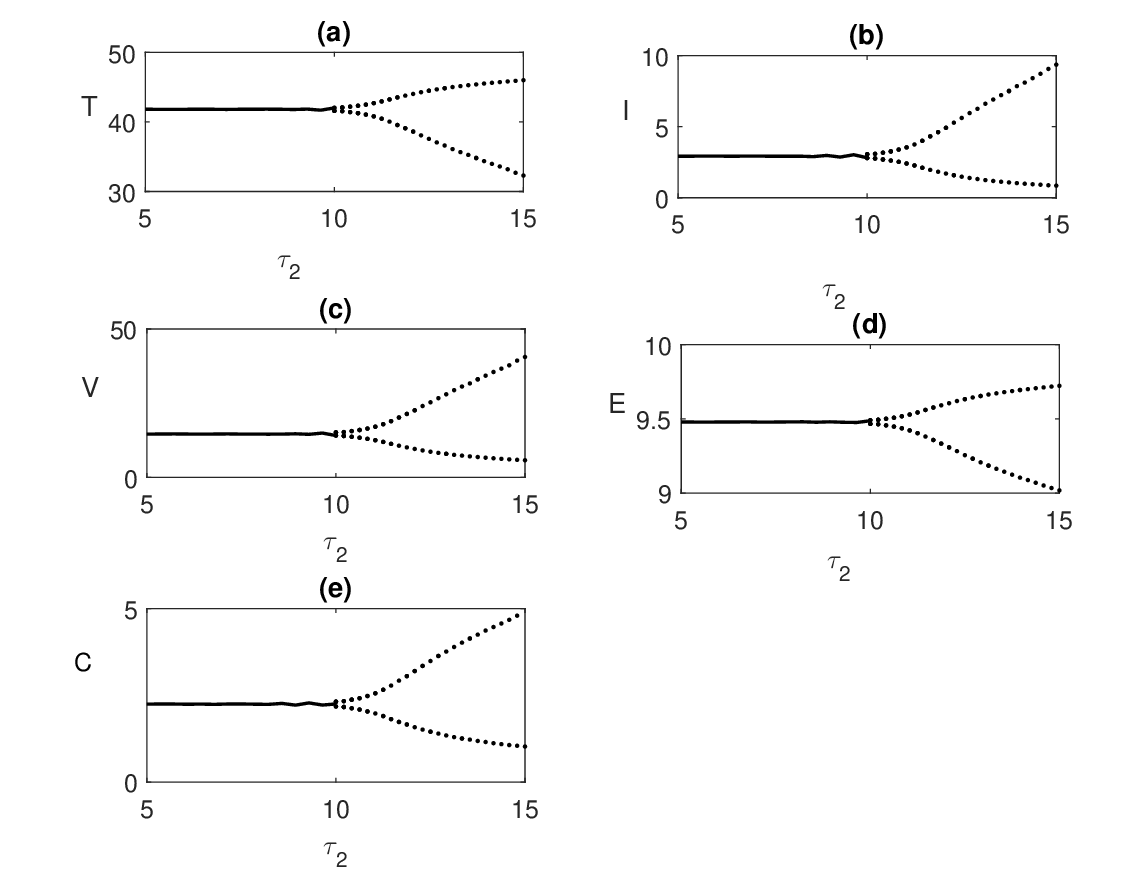}
\caption{Bifurcation diagram of the system taking $\tau_2$
as the bifurcation parameter and $\tau_1=0$. Other parameters
are the same as for Figure~\ref{fig4}.
The solid line indicates the stable endemic equilibrium.}
\label{fig5}
\end{figure}
\begin{figure}[htb!]
\centering
\includegraphics[width=5in,height=5in]{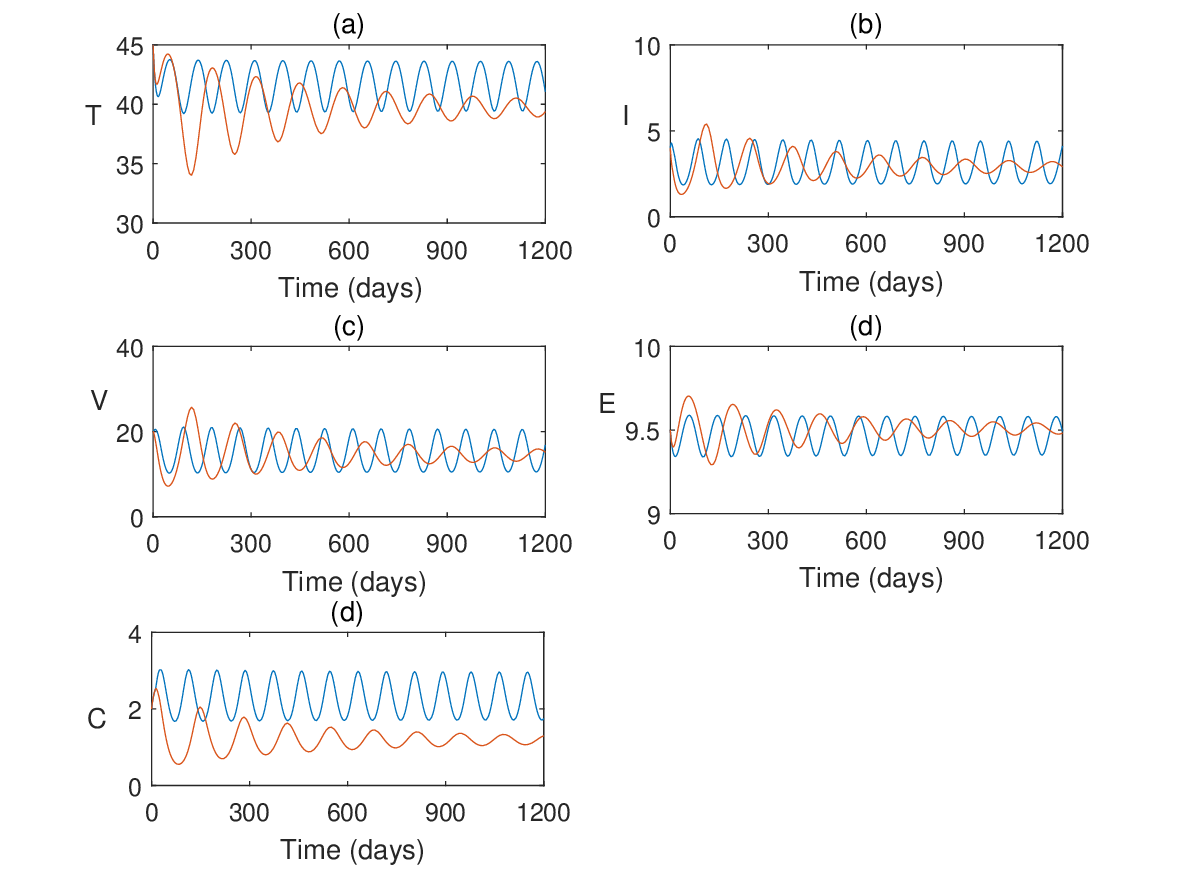}
\caption{Numerical solution of the system with $\tau_2=12$
and for different values of $\tau_1$. The blue line is for
$\tau_1=6$ and the red line is for $\tau_1=1$.}
\label{fig6}
\end{figure}
\begin{figure}[htb!]
\centering
\includegraphics[width=5in,height=5in]{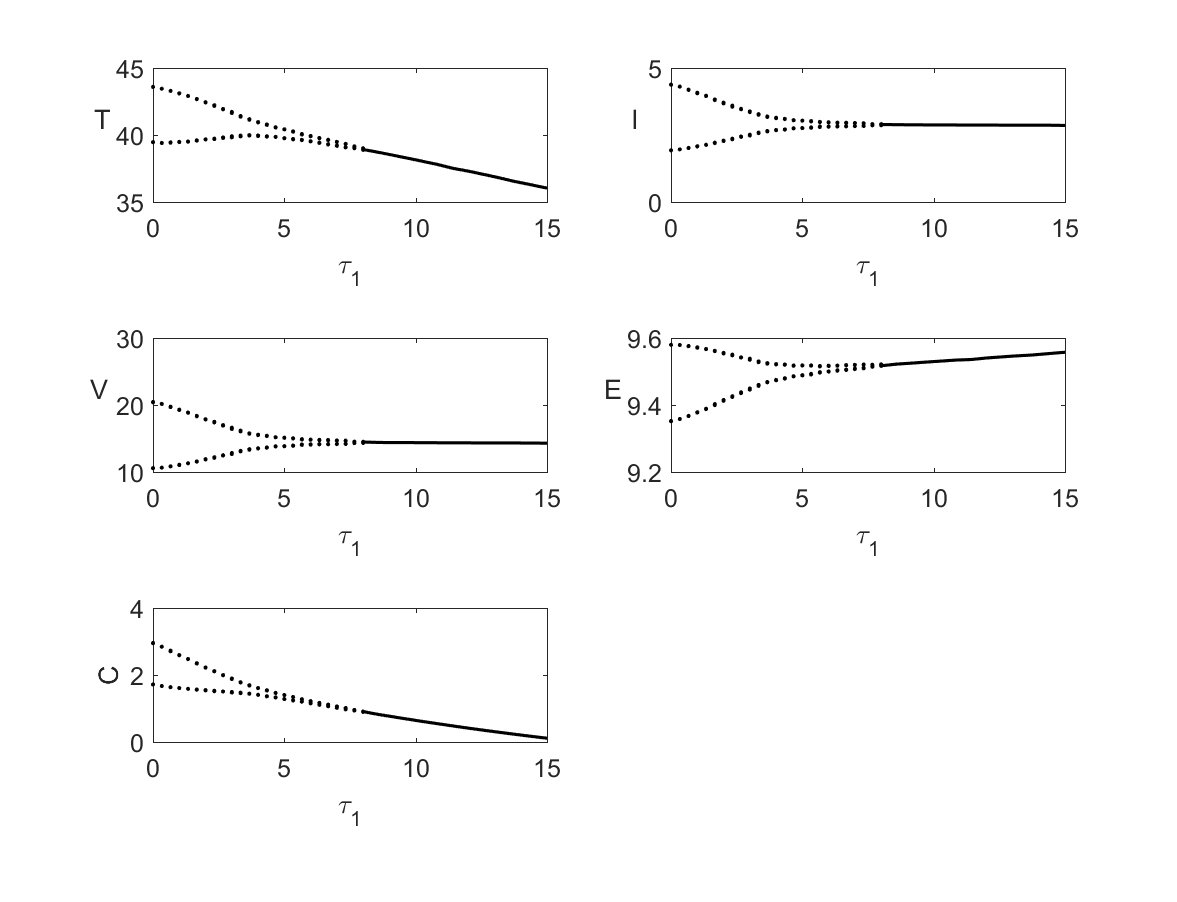}
\caption{Bifurcation diagram of the system taking $\tau_1$
as the bifurcation parameter ($\tau_2=12$). The solid line
indicates the stable endemic equilibrium.}
\label{fig7}
\end{figure}

% -------------------------------------------

\section{Discussion}
\label{sec6}

We have studied the CTL responses on SARS-CoV-2 infected epithelial cells
through mathematical modeling. In our proposed model, we have considered a delay
in the disease transmission term and a CTL proliferation term.
We have studied the delay-induced system both analytically and numerically.

Our analytical findings reveal that the non-delayed system attains
its infection-free state when the basic reproduction number $R_0$ given by \eqref{eq:R0}
is less than 1 and that the system moves to the infected state when $R_0>1$.
We have also witnessed that the infected equilibrium $E^*$ becomes unstable
as the time delay increases. Also, it is observed that as the time delay increases
the system remains unstable or it may cycle through an infinite sequence of regions,
alternately locally asymptotically stable and unstable.
Again, we found the threshold value of the delays
for which the system undergoes Hopf bifurcation.

Numerical findings of the model confirm our analytical results. We found for
the estimated parameter values, as mentioned in Table~\ref{table1}, that
the infected state becomes unstable as the time delay increases but did not
stabilize again at its higher values. In the unstable region, the system shows
an oscillatory behavior at the critical time delay. Our analysis also demonstrates
that the time delay, as well as the parameters $p$ and $q$, play a crucial role
in the stability of the system. The competition between the parameters $p$, $q$,
and the delays to control the system depends on their respective numerical values.

Summarizing, our analytical, as well as numerical findings, tell us that increasing
the immune time delay $\tau_2$ in the system causes always an increasing of oscillations
and instability while, in contrast, the incubation delay $\tau_1$
has a stabilizing role.

% -------------------------------------------

\section{Conclusion}

Time delay models are important tools in the modeling of infectious 
diseases because they allow for more realistic representations of the spread and 
course of infection. In this study, a time delay model of COVID-19 has been proposed 
and analyzed. 'Latent period' and the 'time for immune response' have been considered 
as the time delay parameters for this study. We demonstrated that the delay in immune 
response has a destabilizing effect, whereas latent delay contributes to stabilization. 
The proposed time delay model is more realistic. The outcomes of this study may 
be beneficial for appropriate treatment of COVID-19.

% -------------------------------------------

\section*{Data availability}

The data supporting the findings of this study are available within the article. 

% -------------------------------------------

\section*{Acknowledgements}

Torres was supported by FCT through project UIDB/04106/2020 (CIDMA).
The authors are thankful to a reviewer for useful comments.

% -------------------------------------------

% -------------------------------------------

\end{document}